\documentclass[12pt]{article}
\pdfoutput=1 
\usepackage[T1]{fontenc}
\usepackage{lmodern}
\usepackage[colorlinks=true, pdfstartview=FitV, linkcolor=blue,citecolor=blue, urlcolor=blue]{hyperref}

\usepackage{microtype}
\usepackage[margin=1in]{geometry}
\usepackage[title]{appendix}
\usepackage[round]{natbib}
\usepackage[dvipsnames]{xcolor}
\usepackage{amsmath,amsthm,amssymb,ifthen,mathrsfs,amsfonts,bm,bbm}
\usepackage{multirow}
\usepackage{empheq}
\usepackage{mathtools}
\usepackage{mathabx}
\usepackage{subcaption}
\usepackage{graphicx}
\usepackage{accents}
\usepackage{enumitem}
\usepackage{bigints}
\usepackage{bookmark}
\usepackage{caption}
\usepackage[capposition=below]{floatrow}

\usepackage[framemethod=tikz]{mdframed}
\usepackage{pgf,tikz}
\usetikzlibrary{calc,math}
\usetikzlibrary{arrows,shapes,positioning,automata}
\usetikzlibrary{decorations.markings}
\usetikzlibrary{shapes.geometric}
\usepackage{pgfplots}
\usetikzlibrary{intersections}
\pgfplotsset{compat=newest}
\usepgfplotslibrary{fillbetween}
\usetikzlibrary{patterns}
\usepgfplotslibrary{external}
\newtheorem{theorem}{Theorem}
\newtheorem{lemma}{Lemma}
\newtheorem{proposition}{Proposition}
\theoremstyle{definition}

\newtheorem{definition}{Definition}

\newtheorem{assumption}{Assumption}

\newcommand{\BE}{\mathsf{E}}
\newcommand{\BP}{\mathsf{P}}

\newcommand{\BN}{\mathbb{N}}

\renewcommand{\d}{\mathrm{d}}
\def\ee{\mathrm{e}}
\def\dd{\mathrm{d}}
\newcommand{\ve}{\varepsilon}

\makeatletter
\newcommand{\rn}[1]{\romannumeral #1}
\newcommand{\Rn}[1]{\expandafter\@slowromancap\romannumeral #1@}
\makeatother

\DeclareMathOperator*{\argmax}{arg\,max}

\numberwithin{equation}{section}

\title{The Hazards and Benefits of Condescension in Social Learning}

\author{
  Itai Arieli\thanks{Technion, E-mail: iarieli@technion.ac.il\smallskip}
  \hspace{1cm}
  Yakov Babichenko\thanks{Technion, E-mail: yakovbab@technion.ac.il\smallskip}
  \hspace{1cm}
  Stephan M\"uller\thanks{University of G\"ottingen, E-mail: stephan.mueller@wiwi.uni-goettingen.de\smallskip}
  \hspace{1cm}\\
  Farzad Pourbabaee\thanks{Caltech, E-mail: far@caltech.edu\smallskip}
  \hspace{1cm}
  Omer Tamuz\thanks{Caltech, E-mail: omertamuz@gmail.com\smallskip}
  \hspace{1cm}
}

\begin{document}

\maketitle

\begin{abstract}
  In a misspecified social learning setting, agents are condescending
  if they perceive their peers as having private information that is
  of lower quality than it is in reality. Applying this to a standard
  sequential model, we show that outcomes improve when agents are
  mildly condescending. In contrast, too much condescension leads to
  worse outcomes, as does anti-condescension.
\end{abstract}

\vspace{2cm}

\clearpage
\pagenumbering{arabic}
\setcounter{page}{1}

\interfootnotelinepenalty=10000
\medmuskip= 0.5mu plus 1.0mu minus 1.0mu

\section{Introduction}

Most human decisions are made under uncertainty and in a social
context.  Understanding how economic agents use their private and
social information to form beliefs is a prerequisite for the
understanding of important phenomena such as the diffusion of ideas,
the adoption of technologies, or the formation of political
opinions. In particular, agents' beliefs about their peers'
information is an important factor that can play a decisive role in
the social outcome.

We study the effect of condescension on social learning outcomes: What
happens when agents, through misspecification, underestimate the
quality of the information that their peers have? Our main result is
that condescension can lead to improved social outcomes, as long as it
is mild. In contrast, anti-condescension, in which agents overestimate
their peers' quality of information, leads to bad outcomes, as does
too much condescension.

We study a misspecified version of the classical sequential social
learning model of \cite{Bikhchandani_etal_JPE_1992} and
\cite{Banerjee_QJE_1992}, with unbounded signals, as introduced by
\cite{Smith_Sorensen_ECMA_2000}. Our notion of a good social learning
outcome is that of efficient learning
\citep{Rosenberg_Vieille_ECMA_2019}, which is said to occur when the
number of agents who choose the incorrect action has finite
expectation. In the well-specified setting, the number of agents who
choose the incorrect action is always finite
\citep{Smith_Sorensen_ECMA_2000}, but its expectation can be finite or
infinite \citep{Rosenberg_Vieille_ECMA_2019}.

In our misspecified setting agents perfectly understand and interpret
their own signal but misperceive the quality of their predecessors'
signals. When agents are mildly condescending, efficient learning
occurs. Because agents underestimate the quality of others' signals,
they put too little weight on their predecessors' actions. In
consequence, their actions are suboptimal, but reveal more of their
own private information. When this is done in moderation more is
gained than lost, and in the long run, the result is quick convergence
to the correct action. This occurs even with signal distributions that
would have induced inefficient learning for well-specified agents. Of
course, since agents are misspecified, each agent attains lower
expected utility than they would if they were not, ceteris
paribus. Nevertheless, their behavior has positive externalities on
later agents, with improved asymptotic outcomes.

When agents are too condescending they put so little weight on their
predecessors' actions that no herd forms and both actions are taken
infinitely often, i.e., asymptotic learning is not obtained. When
agents are anti-condescending they put too much weight on their
predecessors' actions. In consequence, wrong herds form with positive
probability, and again asymptotic learning is not
obtained. Interestingly, it follows that asymptotic learning is
equivalent to efficient learning across all misspecified regimes.

Our proof techniques follow those introduced by \cite{hann2018speed}
and \cite{Rosenberg_Vieille_ECMA_2019} who approximate the discrete
time dynamics of the public belief using a continuous time
differential equation. Due to the misspecified nature of our model, our
analysis deviates from theirs in a number of places. For example, we
need to circumvent the fact that the misspecified belief is not a
martingale. In their model, asymptotic learning is guaranteed by this
martingale property \citep{Smith_Sorensen_ECMA_2000}, whereas in our
model we need to prove it by other means.  

\paragraph{Related Literature.}
A closely related paper is \cite{bernardo2001evolution}. They study a cascade setting with binary signals, but where some fraction of the agents are \emph{overconfident}: They do not put enough weight (in Bayesian terms) on the public information. Through mostly numerical analysis, the authors reach a conclusion that is similar to ours: Moderate overconfidence is beneficial for society. 

The empirical literature on social learning supports the idea of overweighting the private information relative to the public information. For example, \cite{Weizsacker_AER_2010} finds in a meta-analysis of 13 social learning experiments that subjects underweight their social observations relative to the payoff-maximizing strategy. \cite{Duffy_etal_EER_2021} find sizable proportions of both behavioral types, i.e. relative over-
and underweighting of the private information.
Condescension provides one possible mechanism which leads to individually suboptimal low weights of public information.  Other mechanisms that may cause a distortion in the optimal weighting of private and public information are, for example, cursedness \citep{Eyster_Rabin_ECONOMETRICA_2005} and na\"ivet\'e \citep{Eyster_Rabin_AEJMicro_2010}.


Our paper is related to the literature on  social learning with
misspecification. Most of this literature documents a detrimental
effect of misspecification on asymptotic outcomes: For example, the
gambler's fallacy leads to incorrect learning almost surely
\citep{He_TE_2022}, and misinterpreting peers' preferences can lead to
incorrect \citep*{Frick_etal_Econometrica_2020} and cyclical
\citep{Gagnon_Bartsch_2016} learning, or entrenched disagreement
\citep{Bohren_Hauser_2019}. \cite{frick2023belief} analyze belief convergence in a general setting. They demonstrate that in the sequential social learning environment arbitrarily small amounts of misspecification can lead to extreme failures of learning.  \cite{Bohren_JET_2016} studies agents with misspecification regarding the correlation between others' actions, and shows that various undesirable outcomes are possible, depending on the degree and direction of misspecification; see also \cite{bohren2021learning} for a more general setting that subsumes a number of previous ones. Our results highlight potential positive
welfare effects, i.e., misspecifications may increase the efficiency
of learning. Similarly to some of these papers, our results show that learning outcomes depend on the behavior of the signal distributions near the fixed points of the learning dynamics, which in our case are the extremal beliefs.

The literature has identified other channels---not requiring misspecification---by which agents may put more weight on their own signals, as compared to the signals of others. For example, idiosyncratic taste shocks \cite[see, e.g.,][]{goeree2006social, lobel2016preferences} imply that an agent's own signal carries more information about their own payoff relevant state than do the signals of the others. As in our model, this causes agents to reveal more of their private information through their actions, which in turn  can improve information  aggregation in the long term.

A paper that is slightly further away in its goal---but closer in techniques and some of the results---is \cite{chen2019sequential}, who studies ambiguity in sequential social learning. In his model, agents have ambiguity about the distributions of the other agents' signals. The main conclusion is that information cascades are a robust outcome that occurs whenever there is sufficient ambiguity.  
Technically, similar observations to our Propositions~\ref{thm:anti-cond} and~\ref{thm:over-cond} appear in \cite{chen2019sequential}, but our paper is focused on the speed of learning (i.e., efficient learning), which is not studied in that paper.

Our work also complements a burgeoning literature which
analyzes the rationale for the persistence of
misspecifications
\citep[e.g.,][]{He_Logober_WP_2020, ba2021robust, Fudenberg_Lanzani_WP_2022}. That
is, in a sequential learning environment misspecified agents might
have an evolutionary advantage over correctly specified agents by
learning the true state of the world faster. Consequently,
misspecifications caused by intermediate levels of condescension might persist in the long run.


\section{Model}
\subsection{Social Learning with Misspecification}
 There is a binary state of nature $\theta\in\Theta =
\{\ell,h\}$, chosen at time zero, and equal to $h$ with
probability $\pi\in (0,1)$. A countably infinite set of agents $=\{1,2,\ldots\}$ arrive
sequentially. Each agent $n$, in turn, takes an action $a_n \in
\{\ell,h\}$, with utility 1 if $a=\theta$ and 0 otherwise. Before
choosing her action, agent $n$ observes her predecessors' actions $I_n
= (a_1,\ldots,a_{n-1})$.

Each agent also observes a private signal $s_n\in S$. Here $S$ is some
measurable set of possible signal realizations. Signals are
independent and identically distributed conditioned on the state. We
denote probabilities by $\BP$, and explicitly write $\BP_\pi$ when we
want to highlight varying values of the prior $\pi$. We further use
the notation $\BP_{h}$ to refer to
$\BP(\cdot\mid\theta=h)$, the probability measure $\BP$
conditional on the realized state being $h$. We define $\BP_{\pi,h}$ analogously.

Let $q_n=\BP(\theta=h|s_n)$ be the (random) \emph{private posterior}:
The belief induced by observing the private signal of agent $n$. By a
standard direct revelation argument, we can assume that $s_n=q_n$,
since $q_n$ is a sufficient statistic for $\theta$ given $s_n$. Denote
by $F_\ell$ and $F_h$ the cumulative distribution functions of $q_n$,
conditioned on $\theta=\ell$ and $\theta=h$, respectively. We define $F
= \frac{1}{2}(F_\ell+F_h)$. This is the cumulative distribution
function of $q_n$, for prior $\pi=1/2$.

So far, this model matches the standard herding model
\citep{Bikhchandani_etal_JPE_1992, Banerjee_QJE_1992,
  Smith_Sorensen_ECMA_2000}. We deviate from these models by introducing
a misspecification regarding others' private signals: Agents correctly observe their own type $q_n$, but have a (common) misspecified prior about the distribution of types. Namely, each
agent believes that all the others' private posteriors have conditional
distributions $\widetilde F_\ell$ and $\widetilde F_h$. Furthermore, it is
common knowledge that these are the agents' beliefs. Note that agents
still interpret their own private signals correctly, with agent $n$
calculating $q_n$ from $s_n$ according to $q_n=\BP(\theta=h|s_n)$. We denote by $\widetilde\BP$ the posterior
probabilities calculated according to the agents' misspecified
beliefs.

In equilibrium, agents choose actions $a_n$ to maximize  their
subjective expected utilities:
\begin{align*}
  a_n = \argmax_{a \in \{\ell,h\}}\widetilde\BP(\theta=a|I_n,q_n).
\end{align*}
We will below restrict ourselves to $q_n$ with non-atomic
distributions, i.e., we assume that $F_\ell$ and $F_h$ are continuous.
This will ensure that agents are never indifferent and the maximum
above is unique. We will likewise assume that $\widetilde F_\ell$ and
$\widetilde F_h$ are continuous.

A pair of conditional CDFs $(F_\ell,F_h)$ is \textit{symmetric}
(around $q=1/2$) if $F_\ell(q)+F_h(1-q)=1$. This in turn implies
$F(q)+F(1-q)=1$. To simplify our exposition we will make the following
assumption.
\begin{assumption}[Symmetry]
  \label{assum: symmetry}
  We assume throughout that $(F_\ell,F_h)$ and $(\widetilde F_\ell, \widetilde
  F_h)$ are symmetric.
\end{assumption}
 
When the prior is $\pi=1/2$, this is equivalent to requiring that the model is invariant with respect to renaming the states.

\subsection{Efficiency}
To study efficiency in this setting, we follow
\cite{Rosenberg_Vieille_ECMA_2019} and introduce some additional
notation. Let $W:=\#\left\{n: a_n\neq\theta \right\}$ be the (random)
number of agents who take the incorrect action.

The next definition includes two notions of efficiency of social
learning.
\begin{definition}
  \begin{enumerate}
  \item \textit{Asymptotic learning} holds if all agents, except
    finitely many, choose the correct action. That is, if $W$ is
    finite $\BP$-almost surely.
    
  \item \textit{Efficient learning} holds if
    $\BE\left[W\right]<\infty$.
    
  \end{enumerate}
\end{definition}
Note that asymptotic learning is equivalent to the sequence of actions $a_n$ converging to $\theta$, which is again equivalent to $a_n=\theta$ for all $n$ large enough. Note also that efficient learning implies asymptotic learning. 

\subsection{The Well-Specified Case}
Without misspecification, the classical herding result of
\cite{Bikhchandani_etal_JPE_1992} is that asymptotic learning does not
hold for any finitely supported private signal distribution in which
no signal is revealing. This is an outcome that displays extreme
inefficiency: With positive probability, all but finitely many agents
choose incorrectly, and in particular there is no asymptotic or
efficient learning.  \cite{Smith_Sorensen_ECMA_2000} show that asymptotic
learning holds if and only if signals are \emph{unbounded}: That is,
if the support of the distribution of the private posteriors $q_n$
includes 0 and 1. Thus, when signals are sufficiently informative, the
extreme inefficiency of the wrong herds of
\cite{Bikhchandani_etal_JPE_1992} is overturned.

Nevertheless, this result left open the possibility that many agents
choose incorrectly before the correct herd arrives. To quantify this
intuition, \cite{sorensen1996rational} gave an example in which
learning is not efficient: $\BE[W]$, the expected number of agents who
choose incorrectly, is infinite. He also conjectured that this is the
case for every signal distribution. This conjecture was shown to be
false by \cite{hann2018speed} and
\cite{Rosenberg_Vieille_ECMA_2019}. In particular,
\cite{Rosenberg_Vieille_ECMA_2019} give an elegant necessary and
sufficient condition for efficient learning, showing that efficient
learning holds if and only if $\int_0^1\frac{1}{F(x)}\dd
x<\infty$.

\subsection{Condescension}
We use our misspecified social learning framework to study how
outcomes change when agents are \emph{condescending}, or think that
others' signals are less informative than they really are. To
formalize and quantify this notion, we restrict ourselves to signals
that are \emph{tail-regular}: A pair of symmetric conditional CDFs
$(F_\ell,F_h)$ is tail-regular if there exists $\alpha>0$ such that
$F(q)=(F_\ell(q)+F_h(q))/2$ behaves like $q^\alpha$ near $q=0$. Formally, if
\begin{equation*}
\begin{gathered}
    0<\liminf_{q \to 0} \frac{F(q)}{q^{\alpha}} \leq \limsup_{q\to 0} \frac{F(q)}{q^{\alpha}}<\infty.
    \end{gathered}    
\end{equation*}
We use Landau notation and write
\begin{align*}
  F(q) = \Theta(q^\alpha)
\end{align*}
as a shorthand for the expression above.\footnote{More generally in
Landau notation, given two functions $f(x)$ and $g(x)$, one writes
$f(x) = \Theta(g(x))$ if
\begin{equation*}
\begin{gathered}
    0<\liminf_{x \to 0} \frac{f(x)}{g(x)} \leq \limsup_{x\to 0} \frac{f(x)}{g(x)}<\infty.
    \end{gathered}    
\end{equation*}
}
\begin{assumption}[Tail-Regularity]
  \label{assum: tail_reg}
  We assume throughout that $( F_\ell,F_h)$ and $(\widetilde F_\ell, \tilde
  F_h)$ are tail-regular.
\end{assumption}

The exponent $\alpha$ associated with a symmetric, tail-regular signal
is unique, and given by
\begin{align*}
  \alpha = \lim_{q \to 0}\frac{\log F(q)}{\log q}.
\end{align*}
It captures a notion of the thinness of the tail of the signals: For
high $\alpha$ there is a small chance of very informative signals, as
compared to low $\alpha$. Thus, in an asymptotic sense, signals are less
informative for higher $\alpha$. Note that by a standard argument
(Lemma~\ref{lem: Psi_h_asymptotics}) if $F(q)=\Theta(q^\alpha)$ then
$F_\ell(q) = \Theta(q^\alpha)$ and $F_h(q)=\Theta(q^{\alpha+1})$.

As an example of tail-regular signals, consider the family of beta distributions, which are commonly used in the applied literature to model the distribution of posterior beliefs \citep[see, e.g.,][]{mckelvey1992experimental, nyarko2006informational, bosch2010finite, ccelen2020belief}. This is a family of probability distributions on the interval $[0,1]$ parametrized by $\alpha,\beta>0$, with probability density function given by $g_{\alpha,\beta}(q) = C q^{\alpha-1}(1-q)^{\beta-1}$, where $C$ is a normalization constant. The parameters $\beta$ and $\alpha$ describe the thickness of the distribution around $0$ and $1$, respectively. Suppose that  private signals have conditional densities $f_\ell = g_{\alpha,\alpha+1}$ and $f_h = g_{\alpha+1,\alpha}$. Then, the unconditional density is again a beta distribution with probability density function $f = g_{\alpha,\alpha}$. This is easily seen to be symmetric and tail-regular, with exponent $\alpha$. 

An important example of signals that are not tail-regular is that of conditionally Gaussian signals, with distributions $\mathcal{N}(-m,1)$ and $\mathcal{N}(+m,1)$, depending on the state. In this case $\lim_{q\to 0} \frac{\log F(q)}{\log q} = \infty$, and the tail-regularity condition is violated, as $F(q)$ decays faster than any polynomial.

In a misspecified model we denote by $\alpha$ and $\tilde \alpha$ the
exponents associated with $(F_\ell,F_h)$ and $(\widetilde F_\ell,\widetilde
F_h)$, respectively. When $\tilde \alpha > \alpha$ we say that agents
are \emph{condescending}: They believe that others' signals are less
informative than they really are. Conversely, when $\tilde\alpha <
\alpha$ agents are \emph{anti-condescending}. Thus,
$\tilde\alpha-\alpha$ is a measure of how condescending the agents are. Note that our definition of condescension is relatively mild, in the sense that it only depends on the tail properties and not on the bulk of the distribution.


\section{Results}
Our first result characterizes the efficiency of learning outcomes for
condescending and anti-condescending agents. The formal proof appears in the appendix. The intuition and dynamics behind this result are presented in detail in Section~\ref{sec:dynamics}.
\begin{theorem}
  \label{thm:main_thm}
  Suppose $\tilde\alpha \neq \alpha$. Then the following are
  equivalent: (\rn{1}) asymptotic learning; (\rn{2}) efficient
  learning;
  (\rn{3}) $\tilde{\alpha}-\alpha \in
  (0,1)$.
\end{theorem}
The regime $\tilde{\alpha}-\alpha \in (0,1)$ describes agents who are
condescending ($\tilde\alpha-\alpha>0$) but not
overly condescending ($\tilde\alpha-\alpha<1$). In this regime the agents' mild condescension causes them to slightly discount the actions of their predecessors, resulting in an increased chance that wrong herds are overturned and a correct herd starts. Since there are infinitely many agents there will be infinitely many fresh chances to do this, and eventually one will succeed, even if each has a very low probability. Moreover, the success probabilities are bounded from below, rendering the expected number of mistakes finite.

In a model without misspecification, the results of
\cite{Rosenberg_Vieille_ECMA_2019} imply that efficient learning
occurs if and only if $\alpha < 1$.  In contrast, Theorem~\ref{thm:main_thm} shows that efficiency can be regained under potentially
small misspecification, for any $\alpha$, as long as agents are
condescending, but not too condescending. Put differently, even for large values of $\alpha$, i.e., when highly informative signals are extremely rare, learning can still be efficient under mild condescension. The technical reason for why a difference of exactly 1 between $\tilde\alpha$ and $\alpha$ is the boundary between mild condescension and over-condescension is related to the summability of certain sequences that determine whether or not a herd can start immediately with positive probability. We explain this in detail in Section~\ref{sec:dynamics}.

Note that the exponent $\alpha$ captures a tail property of the private signals rather than a parameter that determines the entire distribution. Hence two distributions can be very different even if their exponents are very close or the same. When $\tilde\alpha \neq \alpha$, the finer properties of the distributions do not play a role, and efficient learning is solely determined by the tail exponents. When $\tilde\alpha = \alpha$, we conjecture that it is possible for learning to either be efficient or inefficient, depending on finer properties of the distributions; see footnote~\ref{fn:raabe} for the technical details. We leave it for future work to identify these finer properties that determine efficient learning in this regime. 

\medskip

The next two propositions shed light into why learning fails when agents
are either anti-condescending or overly condescending. These results also appear in the supplementary material to \cite{chen2019sequential}. We provide proofs for these claims in Appendix~\ref{app:propositions} for completeness. 

The first proposition concerns the anti-condescension regime, in which  $\tilde{\alpha} < \alpha$.
\begin{proposition}
\label{thm:anti-cond}
  Suppose that $\tilde\alpha <\alpha$. Then, with $\BP$-positive
  probability, a wrong herd forms, i.e., from some point on, all
  agents  take the wrong action.
\end{proposition}
In the case $\tilde\alpha <\alpha$ agents are anti-condescending: They
believe that others have signals that are more informative than they
really are. In consequence, they are more easily swayed by other's
actions, and tend to more often ignore their private signals.%
    \footnote{The effect of anti-condescension is similar to that of na\"ivet\'e in \cite{Eyster_Rabin_AEJMicro_2010}. The mechanism is, however, very different. While na\"ive agents fail to realize that previous movers' also infer from still earlier actions, anti-condescending agents are fully aware of this and take it into account but believe others to have better information than they actually do.}
Thus,
wrong herds can form. This is despite the fact that signals are
unbounded, which, without misspecification, would rule out wrong
herds.

Our next result tackles the question of why learning fails when agents
are overly condescending, i.e., when $\tilde{\alpha} - \alpha \geq 1$.
\begin{proposition}
\label{thm:over-cond}
  Suppose that $\tilde\alpha \geq \alpha+1$. Then, $\BP$-almost surely,
  both actions are taken by infinitely many agents.
\end{proposition}

In the case $\tilde\alpha \geq \alpha + 1$, agents are very
condescending: They think that others have very uninformative
signals. In consequence, they follow their own signals too much, and
herds---wrong or right---do not form: Given enough time, an agent will
come along who will overturn her predecessor's action.


To sum, Proposition~\ref{thm:anti-cond} shows that if
$\tilde \alpha < \alpha$, then when public belief assigns a low probability to the realized state, incorrect cascades remain stable with positive probability. Proposition~\ref{thm:over-cond} shows that if $\tilde\alpha \geq \alpha+1$, then correct cascades are unstable, and thus even if the public belief assigns high probability to the realized state, herds on the correct action almost certainly break down. As we shall show, these conditions are also determinants of efficient learning.
\section{Dynamics}
\label{sec:dynamics}
In this section, we study how agents update their beliefs and choose
their actions under misspecification. We define the public belief and
derive its equations of motion. We show that two properties of the
well-specified model---stationarity and the overturning
principle---still hold in our misspecified environment.

\subsection{Belief Updating}
An important tool in social learning is the public belief (or social
belief) at time $n$:
\begin{align*}
  \pi_n  = \BP\left(\theta=h \big| a_1,\ldots,a_{n-1}\right).
\end{align*}
In our case, however, it is also important to consider  the misspecified
public belief, which is given by
\begin{equation*}
    \tilde{\pi}_n = \widetilde{\BP}\left(\theta=h \big| a_1,\ldots,a_{n-1}\right).
\end{equation*}
The public belief $\pi_n$ is the belief held by a well-specified
observer who sees the agents' actions but not their signals. In
contrast, $\tilde \pi_n$ is the belief held by an observer who holds
the same misspecified beliefs as the agents, and again sees only
actions.

Let $p_n=\BP(\theta=h | I_n,q_n)$ be the posterior belief held by a
well-specified agent who observes all the information available to
agent $n$. The actual, misspecified, posterior of agent $n$ is denoted
$\tilde p_n=\widetilde\BP(\theta=h | I_n,q_n)$. Then by Bayes' Law
\begin{align*}
    \frac{p_n}{1-p_n} &= \frac{\pi_n}{1-\pi_n} \times \frac{q_n}{1-q_n},\\
    \frac{\tilde{p}_n}{1-\tilde{p}_n} &= \frac{\tilde{\pi}_n}{1-\tilde{\pi}_n} \times \frac{q_n}{1-q_n}.
\end{align*}
It follows that the action $a_n$ chosen by agent $n$ is equal to $h$
if $\tilde{\pi}_n+q_n \geq 1$, and to $\ell$ otherwise.\footnote{As we
note above, indifference occurs with probability zero because we
assume that the distribution of $q_n$ is non-atomic.} Thus,
conditioned on $\theta$, the probability that agent $n$ chooses the
low action is $F_\theta(1-\tilde\pi_n)$.

This implies that when agent $n$ chooses the low action, the public
beliefs $\{\pi_n\}$ and $\{\tilde{\pi}_n\}$ evolve as
follows:
\begin{subequations}
  \begin{align}
    \label{eq:rational_belief_l}
    \frac{\pi_{n+1}}{1-\pi_{n+1}} &= \frac{\pi_{n}}{1-\pi_{n}} \times \frac{F_h(1-\tilde{\pi}_n)}{F_\ell(1-\tilde{\pi}_n)}\, ,\\
    \label{eq:misspecified_belief_l}
    \frac{\tilde{\pi}_{n+1}}{1-\tilde{\pi}_{n+1}} &= \frac{\tilde{\pi}_n}{1-\tilde{\pi}_n} \times \frac{\widetilde{F}_h(1-\tilde{\pi}_n)}{\widetilde{F}_\ell(1-\tilde{\pi}_n)}\, .
  \end{align}
\end{subequations}
When agent $n$ chooses the high action,
\begin{subequations}
  \begin{align}
    \label{eq:rational_belief}
    \frac{\pi_{n+1}}{1-\pi_{n+1}} &= \frac{\pi_{n}}{1-\pi_{n}} \times \frac{1-F_h(1-\tilde{\pi}_n)}{1-F_\ell(1-\tilde{\pi}_n)}\, ,\\
    \label{eq:misspecified_belief}
    \frac{\tilde{\pi}_{n+1}}{1-\tilde{\pi}_{n+1}} &= \frac{\tilde{\pi}_n}{1-\tilde{\pi}_n} \times \frac{1-\widetilde{F}_h(1-\tilde{\pi}_n)}{1-\widetilde{F}_\ell(1-\tilde{\pi}_n)}\, .
  \end{align}
\end{subequations}

\subsection{Stationarity and the Overturning Principle}

The above equations of motion imply that as in the well-specified case,
our model is stationary, with $\tilde\pi_n$ capturing all the relevant
information about the past.
\begin{lemma}[Stationarity]
For any fixed sequence $b_1,\ldots,b_k$ of actions in $\{\ell,h\}$ and
any $\tilde\pi \in (0,1)$,
\begin{align*}
  \BP(a_{n+1}=b_1,\ldots,a_{n+k}=b_k\mid \tilde\pi_{n+1}=\tilde\pi) = \BP_{\tilde\pi}(a_{1}=b_1,\ldots,a_{k}=b_k).
\end{align*}
\end{lemma}
That is, suppose that at time $n$ the misspecified public belief
$\tilde\pi_n$ was equal to some $\tilde\pi$. Then the probability that
the subsequent actions are $b_1,\ldots,b_k$ is the same as the
probability of observing this sequence of actions at time $1$, when
the prior is $\tilde\pi$.

Another important observation that generalizes to the misspecified
setting is S{\o}rensen's overturning principle.
\begin{lemma}[Overturning principle]
The misspecified public belief $\tilde{\pi}_{n+1}$ in period $n+1$
is greater than or equal to $1/2$ if and only if $a_n=h$.
\end{lemma}
\begin{proof}
Observe that by the law of total expectation
\begin{equation*}
\begin{split}
    \tilde{\pi}_{n+1} &= \widetilde{\BE}\left[\mathbbm{1}_{\{\theta=h\}} \big|
      a_1,\ldots,a_n\right]\\
    &=  \widetilde{\BE}\Big[\widetilde{\BE}\left[\mathbbm{1}_{\{\theta=h\}} \big| a_1,\ldots,a_n,q_n\right]\Big|a_1,\ldots,a_n\Big]\\
    &= \widetilde{\BE}\left[\tilde{p}_n\big|a_1,\ldots,a_n\right].
\end{split}    
\end{equation*}
Therefore, $a_n=h$ is equivalent to $\tilde{p}_n\geq 1/2$, and hence
equivalent to $\tilde{\pi}_{n+1}\geq 1/2$.
\end{proof}

\subsection{Asymptotic Learning and Immediate Herding}
\label{sec:asympt_learning_immed_herding}
In the misspecified setting the public belief $\tilde\pi_n$ is not a
martingale under the correct measure $\BP$. This martingale property
is an important tool in the proof of asymptotic learning for unbounded
signals in the well-specified case
\citep{Smith_Sorensen_ECMA_2000}. In our case, asymptotic learning
indeed does not always hold, and in particular, we need different tools
to analyze it.

We denote by $a_n \to h$ the asymptotic event that the sequence of
actions converges to $h$. Namely, that $a_n=h$ for all $n$ large enough,
or that a high action herd forms eventually. We denote by $\bar a=h$
the event $\{a_1=h,a_2=h,\ldots\}$ that all agents took the high
action; this is the event that a high action herd formed immediately.

Asymptotic learning occurs when $\BP_\ell(a_n\to \ell)=1$ and
$\BP_h(a_n\to h)=1$. To study the asymptotic events $a_n\to \ell$ and
$a_n\to h$, we study the immediate herding events $\bar a=h$ and $\bar
a=\ell$. These are easier to analyze because conditioned on $\bar
a=\ell$ or on $\bar a=h$, the sequence of misspecified public beliefs
$\{\tilde\pi_n$\} is deterministic and given recursively by
\eqref{eq:misspecified_belief_l} or \eqref{eq:misspecified_belief},
respectively.

To see the connection between asymptotic learning and immediate herding,
condition on $\theta=h$ and consider the event $a_n\to h$ of a good
herd forming eventually. In our setting, we show that this event has
probability 1 if and only if two conditions are met:\footnote{We omit some technical details in the statements of these two conditions. A complete formal treatment is presented in  Appendix~\ref{app:characterization}, Lemmas~\ref{lem:asympt_learning_nec_condition} and~\ref{lem:asympt_learning_suff_condition}.}
\begin{enumerate}[label = (\roman*)]
\item The event $\bar a=\ell$ of an immediate bad herd has probability 0.
\item The event $\bar a=h$ of an immediate good herd has positive
  probability.
\end{enumerate}
The first condition is clearly necessary for asymptotic learning: If
bad herd can form then the probability of a good herd is less than
1. The reason that the second condition is necessary is related to the
stationarity of the process; for a good herd to form eventually, it
must have a positive probability to be formed at any point in time, and hence
also in the beginning. 

To see that these conditions are sufficient for asymptotic learning,
note that again applying stationarity, the first condition implies
that it is impossible for a bad herd to start at any point in
time. This implies that the high action will be taken infinitely
often. Hence, there will be infinitely many chances for a good herd to
start, and thus, by the second condition (and again stationarity) a
good herd will form eventually.

To apply stationarity, we need these two conditions to hold for any
prior, and moreover uniformly so. This is done formally in Appendix~\ref{app:characterization}.

\medskip

Having reduced the problem of asymptotic learning to that of immediate
herding, we turn to calculating the probability of the events $\bar
a=\ell$ and $\bar a =h$. Condition on $\bar a = h$. Then the public
belief $\tilde\pi_n$ evolves deterministically according to
\eqref{eq:misspecified_belief}. It will be useful to consider the misspecified
public log-likelihood ratio $\tilde r_n :=
\log\frac{\tilde\pi_n}{1-\tilde\pi_n}$. In terms of $\tilde r_n$, the
equation of motion \eqref{eq:misspecified_belief} becomes
\begin{align*}
  \tilde{r}_{n+1}=\tilde{r}_n+ \log\frac{1-\widetilde{F}_h\left(\frac{1}{1+\ee^{\tilde{r}_n}}\right)}{1-\widetilde{F}_\ell\left(\frac{1}{1+\ee^{\tilde{r}_n}}\right)}\,.
\end{align*}
It starts at the initial level $\tilde r_1=\log\frac{\pi}{1-\pi}$. When $(\widetilde F_\ell,\widetilde F_h)$ is tail-regular with
exponent $\tilde\alpha$, we can for small $q$ approximate $\widetilde
F_\ell(q)$ with $q^{\tilde\alpha}$ and $\widetilde F_h(q)$ with
$q^{\tilde\alpha+1}$ (neglecting constants). And since
$\tilde r_n$ tends to infinity when $\bar a=h$, this equation of
motion is well approximated by
\begin{align}
    \label{eq:approx-r}
  \tilde{r}_{n+1} \approx \tilde{r}_n+ \ee^{-\tilde \alpha \tilde r_n}\,.
\end{align}
Intuitively, after each observed high action the misspecified public log-likelihood increases by an amount $\ee^{-\tilde \alpha \tilde r_n}$ that
becomes smaller as $\tilde r_n$ increases. More importantly, $\ee^{-\tilde
  \alpha \tilde r_n}$ is also smaller when $\tilde\alpha$ is
higher, i.e., when the signals are less informative: After many high
actions, agents are less surprised to see another high action when
signals are less likely to be very informative.

The asymptotic behavior of this discrete time equation can in turn be
approximated by the differential equation $\frac{\d \tilde r(t)}{\d t}
= \ee^{-\tilde \alpha\tilde r(t)}$ whose solution is $\tilde r(t) =
\tilde\alpha^{-1}\log(1+\tilde\alpha t)$; this is shown formally in Appendix~\ref{app:evolution}, Lemmas~\ref{lem: r_tilde_lower_bound} and~\ref{lem: r_tilde_upper_bound}. Thus, conditioned on the
event $\bar a = h$, the misspecified public log-likelihood $\tilde
r_n$ takes the sequence of deterministic values $\tilde r_n^h$, which
we can approximate by $\tilde r_n^h \approx
\tilde\alpha^{-1}\log(1+\tilde\alpha n)$. Transforming this back to
public beliefs, we get
\begin{align}
  \label{eq:phihn}
  \tilde \pi_n^h \approx 1- n^{-1/\tilde\alpha}.
\end{align}
Thus, the sequence of misspecified public beliefs converges to 1, and it does so
more slowly for higher $\tilde\alpha$, i.e., for less informative signals.

We remind the reader that $a_n=h$ if and only if $q_n \geq 1 - \tilde\pi_n$: The
agent takes the high action if her private posterior
$q_n=\BP(\theta=h|s_n)$ exceeds $1 -\tilde\pi_n$. Hence the event
$\bar a = h$ is the event that $q_n \geq 1-\tilde\pi_n^h$ for all
$n$. Conditioned on $\theta=h$ the (actual, not
misspecified) probability of this event is
\begin{align*}
  1-F_h(1-\tilde\pi_n^h) \approx 1-F_h(n^{-1/\tilde\alpha}) \approx 1-n^{-\frac{\alpha+1}{\tilde\alpha}},
\end{align*}
where the first approximation uses \eqref{eq:phihn} and the second
uses $F_h(q) = \Theta(q^{\alpha+1})$.

Since the random variables $q_n$ are independent conditioned on the state, we
get that the probability of $\bar a = h$ is 
\begin{align*}
  \BP_h(\bar a = h) = \prod_{n=1}^\infty \left(1-F_h(1-\tilde\pi_n^h)\right) \approx
  \prod_{n=1}^\infty \left(1-n^{-\frac{\alpha+1}{\tilde\alpha}}\right).
\end{align*}
Crucially, we are only interested in whether this probability is positive or
zero. As we show formally in Lemma~\ref{lem:S_h_parameter} in Appendix~\ref{app:proof}, the approximations we
perform are good enough, in the sense that the first product vanishes
if and only if the second one does. Thus, by an elementary argument we
get that $\BP_h(\bar a = h)>0$ if and only if $\tilde\alpha-\alpha < 1$.

This argument shows that immediate good herds can form if and only if
$\tilde\alpha-\alpha < 1$, i.e., agents are not
overly condescending. A similar line of reasoning shows that
$\BP_h(\bar a = \ell)=0$ if and only if $\tilde\alpha-\alpha \geq 0$,
i.e., immediate bad herds are excluded when agents are condescending. 

The assumption of tail regularity is used in the approximation $\tilde{r}_{n+1} \approx \tilde{r}_n+ \ee^{-\tilde \alpha \tilde r_n}$ of the evolution of the public log-likelihood ratio, made in \eqref{eq:approx-r}. Tail regularity is also crucial for showing that the differential equation $\frac{\d \tilde r(t)}{\d t}
= \ee^{-\tilde \alpha\tilde r(t)}$ 
 is a good approximation of this discrete time dynamics; see the proofs of Lemmas~\ref{lem: r_tilde_lower_bound} and~\ref{lem: r_tilde_upper_bound} in Appendix~\ref{app:evolution}.

\subsection{Efficient Learning}
\label{subs:efficient_learning}
In the previous section we explained why asymptotic learning holds
only in the regime $\tilde\alpha-\alpha \in [0,1)$. This immediately
implies that outside this range there is also no efficient
learning. In this section we explain why efficient learning does hold
when $\tilde\alpha-\alpha \in (0,1)$.

Suppose $\tilde\alpha-\alpha \in (0,1)$. As asymptotic learning holds,
we know that the agents will take the high action from some point
on. Until then, there will be \emph{runs} of wrong actions, or
sequences of consecutive agents who make the wrong choice. These will
be separated by runs of agents who make the correct choice.

The argument for efficient learning includes two parts. First, we show
that the expected number of bad runs is finite. Second, we show that
the expected length of each bad run is finite. Moreover, the expected
length of a bad run is uniformly bounded, regardless of the history
that came before that run. It follows that the total number of agents
$W$ who take the wrong action has a finite expectation.

The reason that the number of bad runs has finite expectation is that
regardless of the history, there is a uniform lower bound $\delta$ on
the probability that a good herd continues forever. This implies that
the distribution of the number of bad runs is stochastically dominated
by a geometric distribution, which has
a finite expectation. This holds whenever $\tilde\alpha-\alpha < 1$,
i.e., whenever agents are not overly condescending. The argument is
similar to the one from the previous section, which showed that in this
range the probability of $\bar a=h$ is positive in the high state.

To show that the expected length of each bad run is finite, we again
follow the line of argument from the previous section showing that
$\bar a=\ell$ has zero probability in the high state. This holds
whenever $\tilde\alpha > \alpha$, i.e., when agents are
condescending. Moreover, we show that the expected length of a bad run
is uniformly bounded, regardless of the history that came before it
started. This is a consequence of the fact that the public belief at
the onset of a run cannot be arbitrarily high or low, but is bounded
away from 0 and 1. This is a consequence of tail regularity (see Proposition~\ref{prop:finite_exp_duration}).

We note that this last step is obtained in the well-specified setting
of \cite{Rosenberg_Vieille_ECMA_2019} by appealing to the overturning principle and the fact that
$\{\frac{1-\pi_n}{\pi_n}\}$ is a martingale under the correct
conditional measure (that is, $\BP_h$ in the high state).  In
our misspecified case, the overturning principle still holds, but the public 
likelihood $\{\frac{1-\tilde{\pi}_n}{\tilde{\pi}_n}\}$
is not a martingale under the correct measure. Thus, we have
to apply a mechanical method that appeals to tail-regularity.

\subsection{Expected Time of the First Correct Action}

\cite{Rosenberg_Vieille_ECMA_2019} consider another notion of the efficiency of learning, which is briefly discussed in this section. Let $\tau$ be the first time that the correct action is taken:
\begin{align*}
    \tau = \min\{n\,:\, a_n=\theta\}.
\end{align*}
This is a random time that takes values in $\mathbb{N} \cup \{\infty\}$.

\cite{Rosenberg_Vieille_ECMA_2019} show that in the well-specified setting, the finiteness of the expectation of $\tau$ coincides with efficient learning, or the finiteness of the expectation of $W$. 

In our model, when agents are condescending, i.e., when $\tilde\alpha > \alpha$, the expectation of $\tau$ is finite (see Proposition~\ref{prop:finite_exp_duration}). This holds even when agents are over-condescending (i.e., $\tilde\alpha \geq \alpha + 1$), and efficient learning does not hold. In the latter regime, there is no learning because the agents' condescension causes them to put too much weight on their own signals, resulting in both actions being taken infinitely often, and also in small expected $\tau$. When agents are anti-condescending there is a positive probability of $\tau=\infty$ (see Proposition~\ref{prop:low_state_equiv}), and in particular, $\tau$ has an infinite expectation.

\section{Conclusion}

In this paper, we study social learning with condescending agents who underestimate the quality of their peers' information. We show that mild condescension can have positive externalities that result in efficient learning. In particular, there are private signal distributions for which learning is not efficient in the well-specified case, but is efficient with even very small levels of condescension. 

We make several simplifying assumptions for expositional purposes. For example, relaxing symmetry  (Assumption~\ref{assum: symmetry}) yields the same type of results, but where the exponent $\alpha$ needs to be defined for each of the two states (corresponding to the left and right tails), and outcomes can be different in each state. We believe that our results also hold for private belief distributions that are not continuous, but currently, our proof techniques only apply in the continuous case.

A more substantial assumption is that all agents have the same misspecified beliefs about others. We see this a smallest possible deviation from the well-specified case, involving misspecification only about the distribution of agents' types, and nothing else. In particular, because all agents have the same prior,  higher order beliefs are trivial, which makes the model tractable. A natural avenue for future work is to relax this assumption. Indeed, higher order beliefs play an interesting and important role in the misspecified social learning literature \citep{Bohren_JET_2016, bohren2021learning}.

Our analysis of social welfare is restricted to the question of whether the expected number of incorrect actions is finite or not. A more nuanced question is to study how this expectation changes as the actual and perceived distributions of private signals vary. In particular, for private signal distributions where this expectation is finite in the well-specified case, it is interesting to understand how misspecification alters this expectation; this is possible even when $\tilde\alpha = \alpha$. It is furthermore natural to consider a discounted sum of the number of incorrect actions. These are interesting questions that currently seem to be beyond what is technically tractable.

\bibliographystyle{normalstyle}
\bibliography{ref}

\appendix

\section{Preliminaries}
The following lemma is a standard result, with proofs given, for
example, in  Appendix A of \cite{hann2018speed} or
\cite{Rosenberg_Vieille_ECMA_2019}.
\begin{lemma}
\label{lem: NIP}
Let $G_\ell$ and $G_h$ be two cumulative distribution functions
on $[0,1]$, with the Radon-Nikodym derivative $\d G_h/\d G_\ell$
satisfying the iterated likelihood principle $\frac{\d G_h}{\d
  G\ell}(q)=q/(1-q)$. Then it holds that:
\begin{equation*}
    \begin{gathered}
        G_h(q) = 2\left(qG(q)-\int_0^q G(x)\,\d x\right)\,,\\
        G_\ell(q)=2\left((1-q)G(q)+\int_0^q G(x)\,\d x\right)\,.
    \end{gathered}
\end{equation*}
where $G = \frac{1}{2}\left(G_\ell+G_h\right)$. These in turn
imply that $G_h(q) \leq 2qG(q)$ and
$\big|G_\ell(q)-2G(q)\big| \leq 3qG(q)$. Therefore, $\lim_{q
  \to 0}G_h(q)/G(q) =0$ and $\lim_{q \to 0}G_\ell(q)/G(q)=2$.
\end{lemma}

We use this lemma to prove the following additional lemma which
relates the exponent of $G$ to the exponents of $G_\ell$ and $G_h$.
\begin{lemma}
\label{lem: Psi_h_asymptotics}
Suppose $G(q) = \Theta(q^\alpha)$. Then $G_\ell(q)=\Theta(q^{\alpha})$
and $G_h(q)=\Theta(q^{\alpha+1})$.
\end{lemma}
\begin{proof}
Lemma~\ref{lem: NIP} immediately implies that $G_\ell(q) =
\Theta(q^\alpha)$ whenever $G(q) = \Theta(q^\alpha)$. To see that
$G_h(q)=\Theta(q^{\alpha+1})$, note that $G(q) = \Theta(q^\alpha)$
implies there are constants $C\geq c>0$ such that for all $q\in
[0,1]$, one has $cq^\alpha \leq G(q) \leq C q^\alpha$. The previous
lemma thus implies that $G_h(q) \leq 2 C q^{\alpha+1}$. Next, let us
define $m: = \big(c/2C\big)^{1/\alpha}$, and observe that
\begin{equation*}
    G(mq) \leq C(mq)^\alpha = \frac{c}{2}\,q^{\alpha} \leq \frac{1}{2}\,G(q)\,.
\end{equation*}
Since $G$ is increasing, then $G(x) \leq G(q)/2$ for all $x\leq mq$, and therefore, $\int_0^{mq}G(x) \d x \leq m q G(q)/2$. In addition, $G$ being increasing implies that $\int_{mq}^q G(x) \d x \leq (1-m)q G(q)$. Therefore, one obtains the following upper bound for the integral:
\begin{equation*}
    \begin{aligned}
        \int_0^q G(x) \d x &= \int_0^{mq} G(x)\d x + \int_{mq}^q G(x)\d x\\
        &\leq \frac{m}{2}\, q G(q) + (1-m)q G(q) = \left(1-\frac{m}{2}\right)q G(q)\,.
    \end{aligned}
\end{equation*}
It follows from the expression for $G_h$ in the previous lemma that
$G_h(q)\geq m q G(q)$, and hence $G_h(q) \geq m c
q^{\alpha+1}$. Therefore, we have shown that $G_h(q) =
\Theta(q^{\alpha+1})$.
\end{proof}

\section{The Evolution of the Public Log-Likelihood}
\label{app:evolution}

Define the misspecified public log-likelihood ratio by
\begin{align*}
    \tilde{r}_n =\log\frac{\tilde{\pi}_n}{1-\tilde{\pi}_n}\,,
\end{align*}
and the well-specified public log-likelihood ratio by
\begin{align*}
    r_n = \log\frac{\pi_n}{1-\pi_n}\,.
\end{align*}
At $n=1$, it holds that $r_1=\tilde r_1=\log \frac{\pi}{1-\pi}$. Conditioned on the event $\bar a =h$, $\tilde{\pi}_n$ satisfies
the recursive equation \eqref{eq:misspecified_belief}, and thus
$\tilde{\pi}_n$ is deterministic and equals to some $\tilde\pi_n^h$. We
accordingly denote $\tilde r_n^h = \log \frac{\tilde\pi_n^h}{1-\tilde\pi_n^h}$.
\begin{lemma}
  \label{lem: conv_to_one}
  $\lim_{n \to \infty} \tilde\pi^h_n=1$.
\end{lemma}
\begin{proof}
The perceived distributions $\widetilde{F}_h$ and $\widetilde{F}_\ell$ satisfy the iterated likelihood principle, that is,
\begin{equation*}
    \frac{\d\widetilde{F}_h}{\d\widetilde{F}_\ell}(q) = \frac{q}{1-q}\,.
\end{equation*}
This relation implies that $\widetilde{F}_h-\widetilde{F}_\ell$ is strictly decreasing on $[0,1/2]$ and strictly increasing on $[1/2,1]$. Therefore, for every $\pi\in (0,1/2]$ it must be that 
\begin{equation*}
    \widetilde{F}_h(\pi)-\widetilde{F}_\ell(\pi) < \widetilde{F}_h(0)-\widetilde{F}_\ell(0)=0 \Rightarrow \widetilde{F}_h(\pi)<\widetilde{F}_\ell(\pi)\,,
\end{equation*}
and for every $\pi\in [1/2,1)$ one has
\begin{equation*}
    \widetilde{F}_h(\pi)-\widetilde{F}_\ell(\pi) < \widetilde{F}_h(1)-\widetilde{F}_\ell(1)=0 \Rightarrow \widetilde{F}_h(\pi)>\widetilde{F}_\ell(\pi)\,.
\end{equation*}
Observe that due to equation \eqref{eq:misspecified_belief} the sequence $\tilde{\pi}_n$ is strictly increasing. Now, assume by contradiction that $\tilde{\pi}_n \to \hat{\pi} \in (0,1)$, then it must be that 
\begin{equation*}
    \frac{1-\widetilde{F}_h(1-\hat{\pi})}{1-\widetilde{F}_\ell(1-\hat{\pi})}=1\,,
\end{equation*}
which is in contrast with the previous two implications about $\{\widetilde{F}_\ell,\widetilde{F}_h\}$.
\end{proof}

In this section we provide asymptotic results for the evolution of
$r_n$ and $\tilde{r}_n$ on the high action path. As discussed
above, on this path these random variables are deterministic, and
equal to some constants $r_n^h$ and $\tilde r_n^h$,
respectively. These constants satisfy the following reformulation of
expressions in \eqref{eq:rational_belief} and
\eqref{eq:misspecified_belief}:
\begin{align}
  \label{eq: r_difference}
  r_{n+1}^h &= r_n^h + U(\tilde r_n^h)\,,\\
  \label{eq:r_tilde_difference}
  \tilde r_{n+1}^h &= \tilde r_n^h + \widetilde U(\tilde r_n^h)\,,
\end{align}
where
\begin{align*}
  U(r) :=
  \log\frac{1-F_h\left(\frac{1}{1+\ee^r}\right)}{1-F_\ell\left(\frac{1}{1+\ee^r}\right)}\,  ,\\ 
  \widetilde{U}(r):= \log\frac{1-\widetilde{F}_h\left(\frac{1}{1+\ee^r}\right)}{1-\widetilde{F}_\ell\left(\frac{1}{1+\ee^r}\right)}\,.
\end{align*}

One can readily show that both $U$ and $\widetilde{U}$ are decreasing
functions. In addition, they always take positive values, because $F_h
\succeq F_\ell$ and $\widetilde{F}_h \succeq \widetilde{F}_\ell$ in
first order stochastic dominance. This means not only $\{\tilde r^h_n\}$, but also $\{r_n^h\}$ is an increasing sequence.

Lemma~\ref{lem: conv_to_one} implies that $\lim_n \tilde r_n^h =
\infty$. Thus, to study the public belief at large times $n$, we need
to understand $U(r)$ and $\widetilde U(r)$ for large $r$. The next
lemma provides the asymptotic behavior of these functions.
\begin{lemma}
\label{lem: U_F_approx}
For large  $r$, one has $U(r)=\Theta(\ee^{-\alpha r})$ and $\widetilde
U(r)=\Theta(\ee^{-\tilde \alpha r})$, that is
\begin{equation*}
    0< \liminf_{r \to \infty} \frac{U(r)}{\ee^{-\alpha r}} \leq \limsup_{r \to \infty} \frac{U(r)}{\ee^{-\alpha r}} <\infty\,,
\end{equation*}
and
\begin{equation*}
    0< \liminf_{r \to \infty} \frac{\widetilde U(r)}{\ee^{-\tilde \alpha r}}
    \leq \limsup_{r \to \infty} \frac{\widetilde U(r)}{\ee^{-\tilde \alpha r}} <\infty\,.
\end{equation*}
\end{lemma}
\begin{proof}
  Define $\mu = \frac{1}{1+\ee^r}$. We first propose an upper
  bound on $U$. To this end, note that
  \begin{equation*}
    U(r) = \log\frac{1-F_h(\mu)}{1-F_\ell(\mu)} \leq -\log\left(1-F_\ell(\mu)\right)\,.
\end{equation*}
Due to Lemma~\ref{lem: NIP}, $F_\ell(q)\leq 2 F(q)$, therefore $U(r) \leq -\log\left(1-2F(\mu)\right)$. Since for small enough $x$, one has $-\log(1-x)\leq x+x^2$, then
\begin{equation*}
    U(r) \leq 2F(\mu)\left(1+2F(\mu)\right),
\end{equation*}
thereby establishing an upper bound. 

Before proceeding with a lower bound, we introduce the Landau notations, $o(\cdot)$ and $O(\cdot)$: We say $f(x) = o(g(x))$ if $\lim_{x\to 0}\frac{f(x)}{g(x)}=0$, and $f(x)=O(g(x))$ if $\limsup_{x\to 0}\frac{f(x)}{g(x)}<\infty$. 

To propose a lower bound, observe that because of Lemma~\ref{lem: NIP}, $F_h(q)\leq 2qF(q)$ and $F_\ell(q) \geq 2(1-q)F(q)$, therefore,
\begin{equation*}
    \begin{split}
        \ee^{U(r)} &\geq \frac{1-2\mu F(\mu)}{1-2(1-\mu)F(\mu)}\\
        &\geq \big(1-2\mu F(\mu)\big)\big(1+2(1-\mu)F(\mu)+2(1-\mu)^2F(\mu)^2\big)\\
        &=1+2F(\mu)-4\mu F(\mu)+2\big(1-O(\mu)\big)F(\mu)^2\,.
    \end{split}
\end{equation*}
Since $\log(1+x)\geq x-x^2/2$, then 
\begin{equation*}
    \begin{split}
        U(r) &\geq 2F(\mu)-4\mu F(\mu)+\big(2-O(\mu)\big)F(\mu)^2+o\big(F(\mu)^2\big)\\
        &\geq 2F(\mu)\left(1-2\mu+\frac{3}{2}F(\mu)\right).
    \end{split}
\end{equation*}
The above upper and lower bounds imply that $\lim_{r \to \infty}
\frac{U(r)}{2F(\mu)}=1$. In addition, because of tail regularity
(Assumption~\ref{assum: tail_reg}), it holds that $F(\mu) =
\Theta(\ee^{-\alpha r})$, thereby justifying the lemma's first claim. A similar argument implies that $\widetilde{U}(r) = \Theta(\ee^{-\tilde{\alpha}r})$.
\end{proof}

So far our results sidestepped the role of the prior $\pi$, which
determines the initial value for the sequence $\{\tilde{r}^h_n\}$, and
only looked at the asymptotics as $n\to \infty$. In the next lemma, we
establish a property of this sequence, that will prove useful for
uniform convergence results. We use the notation $\tilde{r}_n^h(\pi)$
to refer to the value of $\tilde{r}_n^h$ when the initial belief was
$\pi$, that is, when $\tilde{r}_1^h =
\log\left(\frac{\pi}{1-\pi}\right)$. The rest of the sequence evolves
according to  \eqref{eq:r_tilde_difference}.
\begin{lemma}
\label{lem: uniformity}
For every $\bar{r}\geq 0$, there exists $n_0$ such that $\tilde{r}^h_n(\pi) \geq \bar{r}$ for all $n \geq n_0$ and importantly for all initial $\pi\geq 1/2$.
\end{lemma}
\begin{proof}
The idea is similar to the proof of Lemma~12 in \cite{Rosenberg_Vieille_ECMA_2019}. Let us introduce the mapping $\Psi(r):=r+\widetilde{U}(r)$, and show its $n$-times composition by $\Psi^{n}$. Hence, one has $\tilde{r}^h_n = \Psi^{n-1}(\tilde{r}_1)$. First observe that since $\widetilde{U}>0$, if $\tilde{r}_1 \geq \bar{r}$, then $\Psi^n(\tilde{r}_1) \geq \bar{r}$. Now assume by contradiction that the conclusion of the lemma does not hold. Then, for every $n\in \BN$, there exists an initial belief $\pi^{(n)}$ such that $\tilde{r}^h_n(\pi^{(n)}) \leq \bar{r}$. Also, one has $\Psi^{m-1}\big(\tilde{r}_1(\pi^{(n)})\big)=\tilde{r}^h_m(\pi^{(n)}) \leq \bar{r}$ for all $m\leq n$. Since the interval $[0,\bar{r}]$ is compact, there is a subsequence of initial values $\{\tilde{r}_1(\pi^{(n)})\}$, that we index by $k$, which is converging to $r_* \in [0,\bar{r}]$. Since the mapping $\Psi^{n}$ is continuous for every fixed $n$, then one has
\begin{equation*}
    \Psi^n(r^*) = \lim_{k\to \infty} \Psi^n\big(\tilde{r}_1(\pi^{(k)})\big)\leq \bar{r}\,.
\end{equation*}
The above inequality holds for every $n$, hence it leads to a contradiction, because for every initial prior $\pi>0$, the induced sequence $\{\tilde{r}^h_n\}$ increases to infinity (this follows from Lemma~\ref{lem: conv_to_one}).
\end{proof}

In the next two lemmas we calculate the asymptotic behavior of $\tilde
r_n^h$ and show that $\ee^{\tilde r_n^h}=\Theta(n^{1/\tilde \alpha})$. We
first establish a lower bound for $\tilde{r}^h_n$. We achieve this by
introducing a lower bound for the increments of $\tilde{r}^h_n$ in
\eqref{eq:r_tilde_difference}. We then approximate the
resulting lower envelope with the solution to a differential
equation.
\begin{lemma}[Lower Envelope]
\label{lem: r_tilde_lower_bound}
The misspecified public log-likelihood satisfies
\begin{equation}
\label{eq: r_tilde_lower_bound}
   \liminf_{n\to \infty} \frac{\ee^{\tilde{r}^h_n}}{n^{1/\tilde{\alpha}}}>0\,.
\end{equation}
\end{lemma}
\begin{proof}
By Lemma~\ref{lem: U_F_approx}, there exists $c>0$ such that for all
sufficiently large $n$ (say $n \geq \bar{n}$), one has
$\widetilde{U}(\tilde{r}^h_n)\geq c\,
\ee^{-\tilde{\alpha}\tilde{r}^h_n}$. Additionally, observe that the
mapping $z \mapsto z+c\, \ee^{-\tilde{\alpha} z}$ is increasing for all
sufficiently large $z$ (say $z\geq \bar{z}$). Since $\tilde{r}^h_n \to
\infty$, one can choose $N \geq \bar{n}$ such that $\tilde{r}_N \geq
\bar{z}$. For all $n\geq N$ it holds that
\begin{equation}
\label{eq: r_tilde_lower_increment}
    \tilde{r}^h_{n+1}-\tilde{r}^h_n = \widetilde{U}(\tilde{r}^h_n) \geq c\,\ee^{-\tilde{\alpha}\tilde{r}^h_n}\,.
\end{equation}
We show that this discrete time equation can be bounded from below by the following differential equation:
\begin{equation*}
    \frac{\d z(t)}{\d t} = c\,\ee^{-\tilde{\alpha} z(t)}\,.
\end{equation*}
This equation has the solution form $z(t) =
\tilde{\alpha}^{-1}\log\big(\kappa+c\tilde{\alpha} t\big)$, where the
initial condition parameter $\kappa$ is chosen so that at $n=N$, we have $z(N) =
\tilde{r}^h_{N}$. Next, we inductively show $\tilde{r}^h_n \geq z(n)$ for
all $n\geq N$, which in turn establishes the claim in \eqref{eq:
  r_tilde_lower_bound}. The base step holds by definition. Suppose the
claim also holds at some $n>N$, i.e., $\tilde{r}_n^h\geq z(n)$. Then, observe that because of the
mean value theorem, there exists $t\in [n,n+1]$ such that
\begin{equation*}
    z(n+1)-z(n) = c\, \ee^{-\tilde{\alpha} z(t)} \leq c\, \ee^{-\tilde{\alpha} z(n)}\,,
\end{equation*}
where the inequality follows because $z(t)$ is increasing. Therefore, one has
\begin{equation*}
    \begin{split}
        z(n+1) &\leq z(n)+c\, \ee^{-\tilde{\alpha} z(n)} \\
        &\leq \tilde{r}^h_n+c\, \ee^{-\tilde{\alpha} \tilde{r}^h_n} \leq \tilde{r}^h_{n+1}\,.
    \end{split}
\end{equation*}
The second inequality holds because $z \mapsto z+c\, \ee^{-\tilde{\alpha} z}$ is increasing for $z\geq \bar{z}$, and $z(n) \geq \bar{z}$ for $n\geq N$. The third inequality holds because of \eqref{eq: r_tilde_lower_increment}. This justifies the claim in \eqref{eq: r_tilde_lower_bound}.
\end{proof}

The next lemma posits an upper bound for the increments of
$\tilde{r}^h_n$. Its proof strategy is similar to that of the previous
lemma, with some additional technical considerations.
\begin{lemma}[Upper Envelope]
\label{lem: r_tilde_upper_bound}
The misspecified public log-likelihood satisfies
\begin{equation}
\label{eq: r_tilde_upper_bound}
   \limsup_{n\to \infty} \frac{\ee^{\tilde{r}_n}}{n^{1/\tilde{\alpha}}}<\infty\,.
\end{equation}
\end{lemma}
\begin{proof}
As in the proof of the previous lemma---but changing the direction of
the inequalities---there exists $C>0$ and $\bar{n}$ such that for all
$n \geq \bar{n}$ one has $\widetilde{U}(\tilde{r}^h_n)\leq C
\ee^{-\tilde{\alpha}\tilde{r}^h_n}$. Likewise, the mapping $z \mapsto
z+C \ee^{-\tilde{\alpha} z}$ is increasing for all $z\geq \bar{z}$, and
we can choose $N \geq \bar{n}$ such that $\tilde{r}^h_N \geq
\bar{z}$. Then for all $n\geq N$ it holds that
\begin{equation}
\label{eq: r_tilde_upper_increment}
    \tilde{r}^h_{n+1}-\tilde{r}^h_n = \widetilde{U}(\tilde{r}^h_n) \leq C\ee^{-\tilde{\alpha}\tilde{r}^h_n}\,.
\end{equation}
Take the following differential equation as an upper envelope for the above difference equation:
\begin{equation*}
        \frac{\d z(t)}{\d t} = 2C\ee^{-\tilde{\alpha} z(t)}\,,
\end{equation*}
with the solution form $z(t) = \tilde{\alpha}^{-1}\log\big(\kappa+2C\tilde{\alpha} t\big)$. Observe that for all $\kappa>0$ and $n\geq 1$ one has
\begin{equation}
\label{eq: kappa_tweak}
    2 \ee^{-\tilde{\alpha} z(n+1)} \geq \ee^{-\tilde{\alpha} z(n)}.
\end{equation}
Therefore, one can choose $\kappa$ large enough such that $z(N) \geq \tilde{r}^h_N$. Next, we inductively show $\tilde{r}^h_n \leq z(n)$ for all $n\geq N$, which in turn establishes the claim in \eqref{eq: r_tilde_upper_bound}. The base step holds by definition. Suppose the claim also holds at some $n>N$. Then, observe that because of the mean value theorem, there exists $t\in [n,n+1]$ such that 
\begin{equation*}
    z(n+1)-z(n) = 2C \ee^{-\tilde{\alpha} z(t)} \geq 2C \ee^{-\tilde{\alpha} z(n+1)} \geq C \ee^{-\tilde{\alpha} z(n)}\,,
\end{equation*}
where the first inequality above holds because $z(t)$ is increasing and the second inequality follows from \eqref{eq: kappa_tweak}. Since $z\mapsto z+C \ee^{-\tilde{\alpha} z}$ is increasing for all $z \geq \bar{z}$, and $n\geq N$, then $z(n) \geq \tilde{r}^h_n \geq \bar{z}$ implies that
\begin{equation*}
    z(n+1) \geq z(n)+C\ee^{-\tilde{\alpha} z(n)} \geq \tilde{r}^h_n +C\ee^{-\tilde{\alpha} \tilde{r}^h_n} \geq \tilde{r}^h_{n+1}\,,
\end{equation*}
where the last inequality follows from equation \eqref{eq: r_tilde_upper_increment}. This concludes the induction and thus establishes the asymptotic upper bound for $\ee^{\tilde{r}^h_n}$ in \eqref{eq: r_tilde_upper_bound}.
\end{proof}
The previous two lemmas jointly imply that
$\ee^{\tilde{r}^h_n}=\Theta(n^{1/\tilde{\alpha}})$. Importantly this
holds regardless of the initial belief $\pi$ (i.e., the initial level
$\tilde{r}_1$). Of course, the implied constants may depend on $\pi$.
\section{Characterization of Asymptotic Learning}

\label{app:characterization}
In Section~\ref{sec:asympt_learning_immed_herding} we drew a connection between asymptotic learning and immediate herding. In this section we formalize this, establishing necessary and sufficient conditions for asymptotic learning in terms of immediate herding. Note that the results of this section, Lemmas~\ref{lem:asympt_learning_nec_condition} and~\ref{lem:asympt_learning_suff_condition} do not require tail-regularity and apply more broadly. 

The first lemma states that asymptotic learning in the high state implies that immediate herding on the high action happens with positive probability for some prior $\pi'$, and immediate herding on the low action cannot occur.
\begin{lemma}[Necessary condition]
\label{lem:asympt_learning_nec_condition}
Assume $a_n\to h$, $\BP_h$-almost surely. Then, the following two conditions hold:
\begin{enumerate}[label = (\roman*)]
    \item $\exists\, \pi'<1$ such that $\BP_{\pi',h}\left(\bar a=h\right)>0$,
    \item $\BP_h\left(\bar a=\ell\right)=0$.
\end{enumerate}
\end{lemma}
\begin{proof}
Condition on $\theta=h$, and let $\sigma$ be the random time of the last incorrect action, which has to be finite, because $a_n \to h$. Since $a_\sigma=\ell$, then it must be that $\tilde{\pi}_{\sigma+1}<1/2$, by the overturning principle. Therefore,
\begin{align*}
   1 = \BP_h\left(a_n \to h\right) 
   &= \sum_{k=0}^\infty \BP_h\left(\sigma=k\right)\\
   &= \sum_{k=0}^\infty \BP_h\left(a_m=h \, \forall m>k, \tilde\pi_{k+1}<1/2\right).
\end{align*}
Applying the law of total expectations, this is equal to
\begin{equation}
\label{eq:a_to_h_conv}
\begin{split}
   &= \sum_{k=0}^\infty \BE_h\left[\BP_h\left(a_m=h \,\forall m>k, \tilde\pi_{k+1}<1/2\ |\ \tilde\pi_{k+1}\right)\right]\\
   &= \sum_{k=0}^\infty \BE_h\left[\BP_h\big(a_m=h \,\forall m>k \ | \ \tilde\pi_{k+1}\big)\mathbbm{1}_{\{\tilde\pi_{k+1}<1/2\}}\right]\\
   &=\sum_{k=0}^\infty \BE_h\left[\BP_{\tilde\pi_{k+1},h}\big(\bar a =h\big)\mathbbm{1}_{\{\tilde\pi_{k+1}<1/2\}}\right]\,,
\end{split}
\end{equation}
where the last equality is an application of stationarity. To show (\rn{1}), assume by contradiction that $\BP_{\pi',h}\left(\bar a=h\right)=0$ for every $\pi' \in [0,1)$. Then the right hand side is equal to zero, thereby resulting in a contradiction. 

Since the event $a_n\to h$ is disjoint from the event $\bar a=\ell$, the assumption that the former happens with probability one implies that the latter has probability zero, and thus we have shown (\rn{2}).
\end{proof}

The next lemma shows that asymptotic learning in the high state is implied by uniformly positive probability (over priors at least one half) for immediate herding on the high action, and zero probability (for any prior at least one half) for immediate herding on the low action.
\begin{lemma}[Sufficient condition]
\label{lem:asympt_learning_suff_condition}
The following two conditions are sufficient for $a_n\to h$, $\BP_h$-almost surely:
\begin{enumerate}[label = (\roman*)]
    \item $\inf_{\pi' \geq 1/2} \BP_{\pi',h}\left(\bar a=h\right)>0$,
    \item $\BP_{\pi',h}\left(\bar a=\ell\right)=0$ for all $\pi' > 0$.
\end{enumerate}
\end{lemma}
\begin{proof}
To show asymptotic learning, we first rule out  convergence to the wrong action, that is we claim $\BP_h(a_n \to \ell)=0$. Let $\sigma$ be the last time that agents take the correct action $h$, hence $\tilde{\pi}_{\sigma+1} \geq 1/2$. Then, $a_n\to \ell$ iff $\sigma <\infty$. Therefore, applying the same logic of equation~\eqref{eq:a_to_h_conv} leads to
\begin{equation*}
    \BP_h(a_n \to \ell) = \sum_{k=0}^\infty \BP_h\left(\sigma =k\right) = \sum_{k=0}^\infty \BE_h\left[\BP_{\tilde\pi_{k+1},h}\big(\bar a =\ell\big)\mathbbm{1}_{\{\tilde\pi_{k+1}\geq1/2\}}\right].
\end{equation*}
Since $\BP_{\pi',h}\left(\bar a=\ell\right)=0$ for all $\pi'>0$ the above expression implies that $\BP_h(a_n \to \ell)=0$. 

As we have shown that $a_n$ does not converge to $ \ell$, it follows that the sequence of actions $a_n$ has some number of bad runs:  consecutive agents who take the wrong action, flanked by agents who take the correct action. To show asymptotic learning it suffices to show that the number of bad runs is finite. Let 
\begin{equation*}
    \delta = \inf_{\pi'\geq \frac{1}{2}}\, \BP_{\pi',h}\left(\bar a=h\right)\,.
\end{equation*}
At the end of a bad run the next action is $h$, and so the misspecified public belief is at least $1/2$. Hence, by stationarity, there is a chance of at least $\delta$ of never having another bad run. Since signals are independent conditioned on the state, this implies that the probability of having $m$ bad runs is at most $(1-\delta)^m$. In particular, the probability of infinitely many bad runs is zero.
\end{proof}

\section{Proof of Theorem~\ref{thm:main_thm}}
\label{app:proof}

We divide the proof of Theorem~\ref{thm:main_thm} into two: Proposition~\ref{prop:asymptotic_learning} characterizes asymptotic learning, and Proposition~\ref{prop:efficient_learning} characterizes efficient learning. Jointly, they imply the theorem.

\subsection{Parametric Characterization for Asymptotic Learning}
\label{subs:parametric_asympt_learning}
In this section we characterize the range of condescension where asymptotic learning is achieved.
\begin{proposition}
\label{prop:asymptotic_learning}
The following are equivalent:
\begin{enumerate}[label = (\roman*)]
    \item Asymptotic learning;
    \item $\tilde\alpha-\alpha \in [0,1)$.
\end{enumerate}
\end{proposition}
To prove this statement we leverage the necessary and sufficient conditions found in Lemmas~\ref{lem:asympt_learning_nec_condition} and \ref{lem:asympt_learning_suff_condition}, as well as prove the following two propositions that relate the probability of immediate herding events to the underlying parameters $\alpha$ and $\tilde\alpha$:
\begin{proposition}
\label{prop:high_state_equiv}
The following are equivalent:
\begin{enumerate}[label = (\roman*)]
    \item $\inf\limits_{\pi' \geq 1/2} \BP_{\pi',h}\left(\bar a=h\right)>0$;
    \item $\exists\, \pi'<1$ such that $\BP_{\pi',h}\left(\bar a=h\right)>0$;
    \item $\tilde\alpha-\alpha<1$.
\end{enumerate}
\end{proposition}

\begin{proposition}
\label{prop:low_state_equiv}
The following are equivalent: 
\begin{enumerate}[label = (\roman*)]
\item $\BP_{\pi',h}\left(\bar a=\ell\right)=0$ for all $\pi' > 0$;
\item $\BP_{\pi',h}\left(\bar a=\ell\right)=0$ for some $\pi' < 1$;
\item $\tilde\alpha-\alpha \geq 0$.
\end{enumerate}
\end{proposition}

Proposition~\ref{prop:high_state_equiv} implies that the probability of an immediate good herd is \textit{uniformly} positive (over all initial beliefs larger than $1/2$) if and only if $\tilde\alpha-\alpha<1$, i.e., when agents are not overly condescending. Proposition~\ref{prop:low_state_equiv} claims that the probability of an immediate wrong herd is zero for all positive initial beliefs if and only if $\tilde \alpha-\alpha\geq 0$, i.e., when agents are condescending.

We use these propositions to prove Proposition~\ref{prop:asymptotic_learning}, before proceeding with their proofs.
\begin{proof}[Proof of Proposition~\ref{prop:asymptotic_learning}]
Suppose that asymptotic learning holds, i.e., $a_n \to \theta$, $\BP$-almost surely. Then  $a_n \to h$, $\BP_h$-almost surely. By Lemma~\ref{lem:asympt_learning_nec_condition}, this implies that condition (\rn{2}) of both  Propositions~\ref{prop:high_state_equiv} and~\ref{prop:low_state_equiv} hold. Hence, by these propositions, conditions (\rn{3}) in the two propositions hold, and $\tilde\alpha-\alpha \in [0,1)$.

Suppose that $\tilde\alpha-\alpha \in [0,1)$. Then condition (\rn{3}) of both  Propositions~\ref{prop:high_state_equiv} and~\ref{prop:low_state_equiv} hold. Therefore, condition (\rn{1}) of both propositions hold. Hence, by Lemma~\ref{lem:asympt_learning_suff_condition}, we have asymptotic learning.
\end{proof}

\subsubsection{Proof of Proposition~\ref{prop:high_state_equiv}}
The proof of the first implication, namely (\rn{1}) $\Rightarrow$ (\rn{2}), is immediate. The next lemma establishes the second implication, i.e., (\rn{2}) $\Rightarrow$ (\rn{3}). In fact, it shows a stronger statement. 
\begin{lemma}
\label{lem:S_h_parameter}
The following are equivalent:
\begin{enumerate}[label = (\roman*)]
    \item $\tilde{\alpha}-\alpha<1$;
    \item $\BP_{h}(\bar a = h) > 0$.
\end{enumerate}
\end{lemma}
\begin{proof}
Conditioned on the event $\bar a = h$, the public belief $\tilde \pi_n$ is deterministic and equals  $\tilde\pi_n^h$. Thus the event $\bar a=h$ is equal to the event $\{q_n+\tilde\pi_{n}^h\geq 1,\,\forall n\}$. Since the random variables $q_n$ are independent conditioned on $\theta=h$, we have that
\begin{align}
\label{eq:bar_a_h}
    \BP_{h}(\bar a =h) = \prod_n\BP_{h}(q_n+\tilde\pi_{n}^h\geq 1) = \prod_n\left(1-F_h(1-\tilde\pi_n^h)\right).
\end{align}

This implies that $\BP_h(\bar a=h)>0$ if and only if $-\sum_n \log\left(1-F_h(1-\tilde\pi_n^h)\right)<\infty$. For two sequences $f_n$ and $g_n$, we say $f_n\sim g_n$ if $\frac{f_n}{g_n}\to 1$ as $n\to \infty$. Since $\tilde{\pi}^h_n \to 1$, then $-\log\left(1-F_h(1-\tilde\pi_n^h)\right) \sim F_h(1-\tilde{\pi}^h_n)$, and the previous sum is finite if and only if
\begin{equation*}
    \sum_n F_h(1-\tilde{\pi}^h_n) < \infty\,.
\end{equation*}
Observe that Lemma~\ref{lem: Psi_h_asymptotics} implies that $F_h(q) = \Theta(q^{\alpha+1})$. Also as $n\to \infty$, we have $\ee^{-\tilde{r}^h_n} \sim 1-\tilde{\pi}^h_n$, therefore, the above sum is finite if and only if
\begin{equation}
\label{eq:exp_r_tilde_Ah}
    \sum_n \ee^{-(\alpha+1)\tilde{r}^h_n} < \infty\,.
\end{equation}
Because of the Lemmas~\ref{lem: r_tilde_lower_bound} and \ref{lem: r_tilde_upper_bound}, one has $\ee^{-(\alpha+1)\tilde{r}^h_n} = \Theta \left(n^{-\frac{\alpha+1}{\tilde{\alpha}}}\right)$. Thus, the sum in \eqref{eq:exp_r_tilde_Ah} is finite if and only if $\tilde{\alpha} -\alpha< 1$. 
\end{proof}

The following two lemmas are aimed at proving the third and final implication in Proposition~\ref{prop:high_state_equiv}, that is (\rn{3}) $\Rightarrow$ (\rn{1}). In the first one, we show that the sum in \eqref{eq:exp_r_tilde_Ah} can be made arbitrarily small if the initial value $\tilde{r}_1$ is large enough. Often in the following expressions, we use the notation $\tilde{r}^h_n(r)$ to refer to the process initiated at $\tilde{r}_1=r$. Also, recall our former notation, where we used $\tilde{r}_n^h(\pi)$ to refer to the process initiated at $\tilde r_1=\log \frac{\pi}{1-\pi}$. We use both of these notations interchangeably depending on the context.
\begin{lemma}
\label{lem:sum_made_small}
Assume $\tilde{\alpha}-\alpha<1$. Then for every $\ve>0$, there exists $\bar{r}\geq 0$ such that for all $r \geq \bar{r}$ one has $\sum_{n\geq 0} \ee^{-(\alpha+1)\tilde{r}^h_n(r)} <\ve$.
\end{lemma}
\begin{proof}
We appeal to the idea used in the proof of Lemma~\ref{lem: r_tilde_lower_bound}. Since $\widetilde{U}(r) = \Theta(\ee^{-\tilde{\alpha}r})$, then there exists $c>0$, and correspondingly a threshold $\bar{r}$, such that $\widetilde{U}(r)\geq c\,\ee^{-\tilde{\alpha}r}$ for every $r\geq \bar{r}$ and the mapping $r \mapsto r+c\,\ee^{-\tilde{\alpha} r}$ is increasing on $[\bar{r},\infty)$. In particular, since $\tilde{r}^h_n$ is increasing in $n$, starting the process at any $\tilde{r}_1=r \geq \bar{r}$, implies 
\begin{equation*}
    \widetilde{U}(\tilde{r}^h_n(r))\geq c\, \ee^{-\tilde{\alpha} \tilde{r}^h_n(r)}\,.
\end{equation*}
Next, we recall the continuous time process $z(t)$ such that $z(0)=\bar{r}$, and 
\begin{equation*}
    \frac{d z(t)}{\d t} = c\, \ee^{-\tilde{\alpha} z(t)}\,.
\end{equation*}
The solution to this differential equation takes the form 
\begin{equation*}
    z(t) = \frac{1}{\tilde{\alpha}}\log\Big(\ee^{\tilde{\alpha}\bar{r}}+c\tilde{\alpha} t\Big)\,.
\end{equation*}
Using induction, similar to the one used in Lemma~\ref{lem: r_tilde_lower_bound}, we can show $\tilde{r}_n^h(r) \geq z(n)$ for every initial value $r\geq \bar{r}$. Therefore, for every $r\geq \bar{r}$, it holds that
\begin{equation*}
    \sum_{n\geq 0} \ee^{-(\alpha+1)\tilde{r}^h_n(r)} \leq \sum_{n\geq 0} \ee^{-(\alpha+1)z(n)} = \sum_{n\geq 0} \Big(\ee^{\tilde{\alpha}\bar{r}}+c\tilde{\alpha} n\Big)^{-\frac{\alpha+1}{\tilde{\alpha}}}\,.
\end{equation*}
Since $\alpha+1>\tilde{\alpha}$, for a given $\ve>0$, we can choose $\bar{r}$ large enough, such that the above sum is less than $\ve$.
\end{proof}

Let $\eta \colon [1/2,1] \to [0,1]$,
\begin{align*}
    \eta(\pi) = \BP_{\pi,h}(\bar a=h),
\end{align*}
be the probability of immediate herding on the high action, conditioned on the high state, when the prior is $\pi$. By \eqref{eq:bar_a_h},
\begin{align*}
    \eta(\pi)  = \prod_n \big(1-F_h(1-\tilde{\pi}^h_{n}(\pi))\big).
\end{align*}

\begin{lemma}
\label{lem:unif_convergence}
Assume $\tilde{\alpha}-\alpha<1$. Then $\eta$ is continuous.
\end{lemma}
\begin{proof}
Let
\begin{align}
\label{eq:eta_n}
    \eta_n(\pi) = \BP_{\pi,h}(a_1=h,\ldots,a_n=h) = \prod_{k=1}^n \big(1-F_h(1-\tilde{\pi}^h_{k}(\pi))\big)\,,
\end{align}
be the probability that the first $n$ agents take the high action, conditioned on the high state, when the prior is  $\pi$. By definition, $\eta(\pi) = \lim_n\eta_n(\pi)$. Since the distribution of the private posteriors $q_n$ is non-atomic, each $\eta_n$ is continuous. Thus, we prove that $\eta$ is continuous by showing that $\eta_n$ converges uniformly to $\eta$.

First, Lemma~\ref{lem: Psi_h_asymptotics} implies that $F_h(\ee^{-r}) =\Theta(\ee^{-(\alpha+1)r})$, and hence there exists $C>0$ such that $F_h(\ee^{-r})\leq C\ee^{-(\alpha+1)r}$. Second, because of Lemma~\ref{lem:sum_made_small}, for a given $\ve_1>0$, there exists $\bar{r} \geq 0$ such that for all $r_1\geq \bar{r}$, one has
\begin{equation*}
   \sum_{n\geq 0} \ee^{-(\alpha+1)\tilde{r}^h_n(r_1)} \leq \ve_1\,.
\end{equation*}
Then, because of Lemma~\ref{lem: uniformity} there exists $n_0\equiv n_0(\bar{r})$ such that $\tilde{r}^h_n(\pi) \geq \bar{r}$ for all initial $\pi\geq 1/2$, and $n\geq n_0$. By \eqref{eq:eta_n}, 
\begin{align*}
    \eta_{n+1}(\pi) = \eta_n(\pi)\big(1-F_h(1-\tilde\pi^h_n(\pi)\big),
\end{align*}
so that $|\eta_{n+1}(\pi) - \eta_n(\pi)| \leq F_h(1-\tilde\pi^h_n(\pi))$. Hence, 
for every $k>0$, and $\pi \geq 1/2$,
\begin{equation}
    \begin{gathered}
    \left|\eta_{n_0+k}(\pi)-\eta_{n_0}(\pi)\right| \leq \sum_{n=n_0}^\infty F_h\big(\ee^{-\tilde{r}^h_n(\pi)}\big)
    \leq C\sum_{n\geq n_0} \ee^{-(\alpha+1)\tilde{r}^h_n(\pi)} 
    \leq C\ve_1\,.
    \end{gathered}
\end{equation}
The third inequality above holds because $\tilde{r}^h_{n_0}(\pi) \geq \bar{r}$, and thus Lemma~\ref{lem:sum_made_small} implies the sum is smaller than $\ve_1$.
Since $\ve_1$ was chosen independently, the final term above can be made arbitrarily small, by taking $n_0$ large enough. This implies the sequence $\{\eta_n\}$ is Cauchy w.r.t.\ the sup-norm in $C[1/2,1]$, and thus it converges uniformly to $\eta$. Therefore, $\eta$ is continuous.
\end{proof}
Using the above lemma, we can now conclude the proof of the last implication in Proposition~\ref{prop:high_state_equiv}, namely (\rn{3}) $\Rightarrow$ (\rn{1}). Assume by contradiction, that condition (\rn{1}) does not hold, then $\inf_{\pi \geq 1/2} \eta(\pi)=0$. By the previous lemma $\eta$ is a continuous function, hence there must exist $\hat{\pi}\in [1/2,1]$ such that $\eta(\hat{\pi})=0$. Since, $\eta(1)\neq 0$, then $\hat{\pi}\in [1/2,1)$ and Lemma~\ref{lem:S_h_parameter} implies that $\tilde{\alpha}-\alpha \geq 1$. This violates the initial assumption (i.e., $\tilde{\alpha}-\alpha < 1$) and hence concludes the proof of Proposition~\ref{prop:high_state_equiv}.
\subsubsection{Proof of Proposition~\ref{prop:low_state_equiv}}

The first implication, namely (\rn{1}) $\Rightarrow$ (\rn{2}), is immediate. For the remaining two implications, define 
\begin{equation*}
    \xi(\pi) = \BP_{\pi,h}\big(\bar a = \ell\big)\,,
\end{equation*}
and appeal to the next lemma.
\begin{lemma}
\label{lem:S_ell_parameter}
For every $\pi \in (0,1)$, one has $\xi(\pi)=0$ if and only if $\tilde{\alpha}-\alpha \geq 0$.
\end{lemma}
\begin{proof}
Because of the symmetry assumption, we have $\xi(\pi) = \BP_{1-\pi,\ell}\left(\bar a = h\right)$. Let $\tilde{\pi}^h_n=\tilde{\pi}^h_n(1-\pi)$ be the misspecified public belief on the high action path, initiated at $\tilde{\pi}_1=1-\pi$. Then, following the same argument of Lemma~\ref{lem:S_h_parameter}, one has
\begin{equation*}
\begin{split}
    \xi(\pi) = \prod_n \big(1-F_\ell(1-\tilde{\pi}^h_{n}(1-\pi))\big)\,.
\end{split}
\end{equation*}
Therefore, $\xi(\pi)>0$ if and only if $-\sum_n \log\big(1-F_\ell(1-\tilde{\pi}^h_n)\big)<\infty$. Since on the high action path $\tilde{\pi}^h_n \to 1$, then $-\log\big(1-F_\ell(1-\tilde{\pi}^h_n)\big) \sim F_\ell(1-\tilde{\pi}^h_n)$, and the previous sum is finite if and only if
\begin{equation*}
    \sum_n F_\ell\left(1-\tilde{\pi}^h_n\right) < \infty\,.
\end{equation*}
Lemma~\ref{lem: NIP} implies that $F_\ell(q)=\Theta(q^\alpha)$. Also, as $n\to \infty$, we have $1-\tilde{\pi}^h_n \sim \ee^{-\tilde{r}^h_n}$, therefore the above sum is finite if and only if
\begin{equation*}
    \sum_n \ee^{-\alpha\tilde{r}^h_n} < \infty\,.
\end{equation*}
It was shown in Lemmas~\ref{lem: r_tilde_lower_bound} and \ref{lem: r_tilde_upper_bound} that $\ee^{\tilde{r}_n}=\Theta(n^{1/\tilde{\alpha}})$, thus one can deduce that the above sum is finite if and only if $\alpha > \tilde{\alpha}$. Therefore, $\xi(\pi)=0$ if and only if $\tilde{\alpha} - \alpha \geq 0$.
\end{proof}
Observe that $\xi(1)=0$. Therefore, the second and the third implications of Proposition~\ref{prop:low_state_equiv}, namely (\rn{2}) $\Rightarrow$ (\rn{3}) $\Rightarrow$ (\rn{1}), respectively follow from the above lemma, thereby concluding the proof of Proposition~\ref{prop:low_state_equiv}.


\subsection{Parametric Characterization for Efficient Learning}

In this section we characterize the range of condescension where efficient learning is achieved.
\begin{proposition}
\label{prop:efficient_learning}
Assume $\tilde\alpha \neq \alpha$. The following are equivalent:
\begin{enumerate}[label = (\roman*)]
    \item Efficient learning;
    \item $\tilde\alpha-\alpha \in (0,1)$.
\end{enumerate}
\end{proposition}
\begin{proof}
The implication (\rn{1}) $\Rightarrow$ (\rn{2}) follows immediately from Proposition~\ref{prop:asymptotic_learning}, since efficient learning implies asymptotic learning.

Towards sufficiency, assume $\tilde\alpha-\alpha <1$. Following the same logic as in the proof of Lemma~\ref{lem:asympt_learning_suff_condition}, one obtains that conditioned on $\theta=h$, the probability of having $m$ bad runs is at most $(1-\delta)^m$ for some $\delta>0$, and hence the number of bad runs has a finite expectation.

By Proposition~\ref{prop:finite_exp_duration} below, conditioned on the high state, $\tilde\alpha-\alpha >0$ implies that the expected length of the first bad run is bounded by $C_0\frac{1-\pi}{\pi}$, for some constant $C_0>0$. This proposition also implies, by stationarity, that conditioned on $\theta=h$ and on any prior history, the expected length of any future bad run is at most $C_0 B$, where $B>0$ is another constant. It thus follows from the fact that signals are conditionally i.i.d.\ that the expected total number of low actions in the high state is finite. The argument is analogous to the one that appears in Appendix~B.3 of \cite{Rosenberg_Vieille_ECMA_2019}.

Finally, by symmetry, the expected number of high actions in the low state is also finite, and thus we have efficient learning.

\end{proof}

We end this second with the following proposition, which is the main ingredient of the proof above. It shows that  $\tilde\alpha-\alpha > 0$ implies that the expected length of a bad run is uniformly bounded. 

Define $\tau_\theta:= \min\{n: a_n=\theta\}$. Note that conditioned on $\theta=h$, $\tau_h$ is the length of the first bad run.
\begin{proposition}
\label{prop:finite_exp_duration}
Assume $\tilde{\alpha}-\alpha >0$, then the following statements hold:
\begin{enumerate}[label = (\roman*)]
    \item Let $\pi\leq 1/2$. There exists a constant $C_0>0$ (independent of $\pi$) such that
    \begin{equation*}
        \BE_{\pi,h}\left[\tau_h\right] \leq C_0\, \frac{1-\pi}{\pi}\,.
    \end{equation*}
    \item Let $\tilde{\pi}_{n+1}$ be the misspecified public belief after observing a history ending with $a_{n-1}=h$ and $a_n=\ell$. Then $\frac{1-\tilde{\pi}_{n+1}}{\tilde{\pi}_{n+1}} \leq B$ for some constant $B<\infty$ that does not depend on the history.
\end{enumerate}
\end{proposition}

\begin{proof}
To see  (\rn{1}), observe that because of symmetry, one has $\BE_{\pi,h}\left[\tau_h\right]=\BE_{1-\pi,\ell}\left[\tau_\ell\right]$. Also, it holds that
\begin{equation*}
    \BE_{1-\pi,\ell}\left[\tau_\ell\right] = 1+\sum_{n\geq 1}\BP_{1-\pi,\ell}\left(\tau_\ell > n\right)\,.
\end{equation*}
By Bayes Law
\begin{equation*}
    \BP_{1-\pi,\ell}\left(\tau_\ell > n\right) = \frac{1-\pi}{\pi} \times \frac{1-\pi_{n+1}^h}{\pi_{n+1}^h} \, \BP_{1-\pi,h}\left(a_1=\cdots=a_n=h\right)\,,
\end{equation*}
where $\pi^h_n$ is the correctly specified public belief on the high action path, starting at $1-\pi$ and following the dynamics in equation \eqref{eq:rational_belief}. Recall that $r_n^h$ represents the correctly specified public log-likelihood on the high action path, which follows the dynamics in equation \eqref{eq: r_difference}, namely $r^h_{n+1}-r^h_n=U(\tilde{r}^h_n)$. Hence, the above expression implies that
\begin{equation}
\label{eq: tau_ell_bound}
    \BE_{1-\pi,\ell}\left[\tau_\ell\right] \leq 1 + \frac{1-\pi}{\pi}\, \sum_{n\geq 1} \ee^{-r^h_{n+1}}\,.
\end{equation}
Since the initial belief is set to $1-\pi$ and it is assumed in part (\rn{1}) that $\pi \leq 1/2$, then $r_n^h \geq 0$.
Next, observe that on the path $\bar a=h$, the misspecified public log-likelihood follows the difference equation \eqref{eq:r_tilde_difference}, namely $\tilde{r}^h_{n+1}-\tilde{r}^h_n=\widetilde{U}(\tilde{r}^h_n)$. Additionally, because of Lemma~\ref{lem: U_F_approx}, there exists $C>0$ such that $\widetilde{U}(\tilde{r}^h_n) \leq C \ee^{-\tilde{\alpha}\tilde{r}^h_n}$. We continue by finding a continuous time upper envelope for $\tilde{r}^h_n$---analogous to Lemma~\ref{lem: r_tilde_upper_bound} with a slight catch in selecting the initial condition. Choose $\bar{r}>0$ such that the mapping $r\mapsto r+C\ee^{-\tilde{\alpha}r}$ becomes increasing on $[\bar{r},\infty)$. Let $n_0:=\min\{n: \tilde{r}^h_n \geq \bar{r}\}$ that is finite because $\tilde{r}_n \to \infty$ on the high action path. Since $\widetilde{U}(\cdot)$ is a decreasing function, then $\bar{r}\leq \tilde{r}^h_{n_0}\leq \bar{r}+\widetilde{U}(0)$. Let $z(t)$ be the solution to the following differential equation
\begin{equation*}
    \frac{\d z(t)}{\d t} = 2C \ee^{-\tilde{\alpha}z(t)}\,,
\end{equation*}
starting at $z(0)=\bar{r}+\widetilde{U}(0)$. Therefore, $z(t)=\tilde{\alpha}^{-1}\log\left(\ee^{\bar{r}+\widetilde{U}(0)}+2C \tilde{\alpha}t\right)$. Following the recipe of Lemma~\ref{lem: r_tilde_upper_bound}, one can show by induction that $z(k) \geq \tilde{r}^h_{k_0+n}$ for all $k\geq 0$.
Next, we examine Raabe's criterion\footnote{\label{fn:raabe} Raabe's criterion for convergence of sums states that $\sum_n \nu_n$ converges if $\liminf_n \rho_n>1$ and diverges if $\limsup_n \rho_n<1$, where $\rho_n = n\big(\nu_n/\nu_{n+1}-1\big)$. It is inconclusive when $\lim_n\rho_n=1$. This latter case corresponds to $\tilde\alpha=\alpha$.} for the infinite sum $\sum_{n\geq 1} \ee^{-r^h_{n}}$, that is to examine the limit of the following expression:
\begin{equation*}
    n\left(\frac{\ee^{-r^h_{n}}}{\ee^{-r^h_{n+1}}}-1\right) = n\left(\ee^{U(\tilde{r}^h_n)}-1\right) \geq n U(\tilde{r}^h_n)\,.
\end{equation*}
Note that $U$ is decreasing, that $z(k) \geq \tilde{r}^h_{k_0+n}$ for all $k\geq 0$, and that there exists $c>0$ such that $U(z) \geq c\, \ee^{-\alpha z}$. Hence
\begin{equation*}
    \liminf_{n \to \infty} n\left(\frac{\ee^{-r^h_{n}}}{\ee^{-r^h_{n+1}}}-1\right) \geq \limsup_{k \to \infty} c k \, \ee^{-\alpha z(k)}=\limsup_{k \to \infty} \frac{c\, k}{\left(\ee^{\bar{r}+\widetilde{U}(0)}+2C \tilde{\alpha}\,k\right)^{\alpha/\tilde{\alpha}}}\,.
\end{equation*}
Since $\tilde{\alpha}>\alpha$, the limit superior on the right hand side above is infinite and thus the sum $\sum_{n\geq 1} \ee^{-r^h_{n}}$ is convergent. Together with \eqref{eq: tau_ell_bound}, this implies that there exists a constant $C_0>0$ such that $\BE_{\pi,h}[\tau_h] \leq C_0 \, \left(\frac{1-\pi}{\pi}\right)$. This establishes (\rn{1}).

To see (\rn{2}), condition on the event $\{a_{n-1}=h, a_{n}=\ell\}$. Equivalently, $\tilde\pi_{n} \geq 1/2$ and $\tilde\pi_{n+1} \leq 1/2$, by the overturning principle. Then, Bayes law implies
\begin{equation*}
    \frac{1-\tilde{\pi}_{n+1}}{\tilde{\pi}_{n+1}} = \frac{1-\tilde{\pi}_n}{\tilde{\pi}_n} \times \frac{\widetilde{F}_\ell(1-\tilde{\pi}_n)}{\widetilde{F}_h(1-\tilde{\pi}_n)}\,.
\end{equation*}
Since $\widetilde{F}_\ell(q) = \Theta(q^{\tilde{\alpha}})$ and $\widetilde{F}_h(q) = \Theta(q^{\tilde{\alpha}+1})$, then, there are constants $C>0$ and $c>0$ such that $\widetilde{F}_\ell(q) \leq C q^{\tilde{\alpha}}$ and $\widetilde{F}_h(q) \geq c q^{\tilde{\alpha}+1}$ for all $q\in [0,1/2]$. Therefore, 
\begin{equation*}
    \frac{1-\tilde{\pi}_{n+1}}{\tilde{\pi}_{n+1}} \leq \frac{1-\tilde{\pi}_n}{\tilde{\pi}_n} \, \frac{C(1-\tilde{\pi}_n)^{\tilde{\alpha}}}{c(1-\tilde{\pi}_n)^{\tilde{\alpha}+1}} \leq \frac{2C}{c}\,.
\end{equation*}
This establishes  (\rn{2}).
\end{proof}

\section{Proofs of Propositions~\ref{thm:anti-cond} and~\ref{thm:over-cond}}
\label{app:propositions}
\begin{proof}[Proof of Proposition~\ref{thm:anti-cond}]

Suppose that $\tilde \alpha < \alpha$. Then, by Proposition~\ref{prop:low_state_equiv} one has $\BP_h(\bar a = \ell)>0$, so that a wrong herd forms immediately with positive probability.
\end{proof}

\begin{proof}[Proof of Proposition~\ref{thm:over-cond}]

Suppose that $\tilde \alpha \geq \alpha + 1$. Condition on the high state. Then, by Proposition~\ref{prop:high_state_equiv}, for any prior $\pi'<1$, the probability of an immediate herd on the high action is zero. Hence, by stationarity, the probability that $a_n \to h$ is zero. By Proposition~\ref{prop:low_state_equiv}, the probability of an immediate herd on the low action is also zero, and hence, again by stationarity, the probability that $a_n \to \ell$ is zero. Thus the agents take both actions infinitely many times. The same argument applies when conditioning on the low state.
\end{proof}
\end{document}